\documentclass[lettersize,journal]{IEEEtran}
\usepackage{amsmath,amsfonts}
\usepackage{algorithmic}
\usepackage{amsmath}
\usepackage{amsthm}
\usepackage{graphicx} 
\usepackage{float}
\usepackage{comment}
\usepackage[natbib=true, style=numeric,sorting=none,maxnames=99]{biblatex}
\usepackage{color}
\usepackage[colorinlistoftodos]{todonotes} 
\usepackage{soul}
\usepackage{algorithm}
\usepackage[caption=false,font=normalsize,labelfont=sf,textfont=sf]{subfig}
\usepackage{textcomp}
\usepackage{stfloats}
\usepackage{url}
\usepackage{color}
\usepackage{verbatim}
\usepackage{graphicx}
\AtBeginBibliography{\small}
\newtheorem{proposition}{Proposition}[section]
\newtheorem{cor}{Corollary}[section]
\newtheorem{lemma}{Lemma}[section]
\allowdisplaybreaks

\begin{document}

\title{A Risk-Based Equilibrium Analysis of Energy Imbalance Reserve in Day-Ahead Electricity Markets}

\author{Ryan Ent, Golbon Zakeri, Jinye Zhao,~\IEEEmembership{Member,~IEEE}, Tongxin Zheng~\IEEEmembership{Fellow,~IEEE}
\thanks{Jinye Zhao and Tongxin Zheng are with the Independent System Operator of New England Inc., Holyoke, MA 01040 USA}
\thanks{Ryan Ent and Golbon Zakeri are with the Department of Industrial Engineering and Operations Research at the University of Massachusetts, Amherst, MA 01003} }
\maketitle
\begin{abstract}
Energy imbalance reserve (EIR) product is a new product in the Independent System Operator (ISO) of New England's day-ahead wholesale electricity market. Different from existing forward reserve products, EIR is a novel real option product, which is settled against the real-time energy price rather than reserve prices. This novel product has not been analyzed in the research literature in terms of its effects. In this paper, we develop a stochastic long-run equilibrium model that incorporates the risk preference of  generator and demand agents participating in the energy and reserve market in both day-ahead and real-time time frame. In a risk neutral environment, we find that the presence of the EIR product makes little difference on market outcomes. We also conduct a series of numerical simulations  with risk-averse generators and demand, and observed increased advanced fuel procurement when the EIR product is present.  
\end{abstract}
\section{Background}
\IEEEPARstart{N}ew England's power system has been relying on natural gas-fired generating units. In fact, 55\% of electricity in 2024 was produced from natural gas \cite{isoresourcemix}. Natural gas is primarily delivered to New England via interstate pipelines, making it especially vulnerable during winter cold spells when residential heating demand surges. During cold spells, gas supply is often constrained and New England gas power plants that rely on just-in-time delivery of natural gas face challenges of procuring fuel in real time. These gas plants often do not have sufficient incentives to make advance arrangements with gas suppliers if they do not expect to operate in real time (i.e. do not have any day-ahead market awards). Of all the ISO regions in the US, ISO New England (ISO NE)'s gas power fleet have the second lowest rate of firm contracts, with 45\% of plants purchasing firm contracts in 2024 \cite{eia923}. Instead, they opt for interruptible contracts on the spot market, which puts them at the lowest priority for fuel delivery when the natural gas infrastructure is constrained. \footnote{In the models presented in Section II, we assume that generators are able to choose between investing in fuel at a single fixed price in advance of real time generation or purchase fuel with a range of different prices that correlate with real time electricity demand on a spot market. The choice between fixed and uncertain fuel costs are not identical to the choice between interruptible or firm fuel contracts, but serve as a proxy for this investment choice on a shorter time scale.}

    In addition, there is frequently a gap between the energy cleared in the day-ahead market (DAM) and the real-time load. In 2023, this gap was positive in 64\% of all hours in 2023, with an hourly average value of 515.7 MW and a maximum value of 3397.9 MW during hours where the energy cleared in the DAM is less than the real-time load \cite{energygapdataset}. Similar persistent real time and day-ahead imbalances exist in other US ISO/RTO regions, even prompting the California system operator (CAISO) to create a new imbalance reserve product \cite{caiso_ir}, \cite{california_iso_day_2023}. Because additional resources needed to meet the gap are often procured after the DAM clears through the Reserve Adequacy Analysis (RAA) process, market clearing prices in the DAM do not reflect the actual cost of meeting real-time load, adversely affecting market transparency. More importantly, participants do not have strong incentives to procure fuel in advance if they do not have any DAM obligations.   

\begin{figure}
    \centering
    \includegraphics[scale = 0.6]{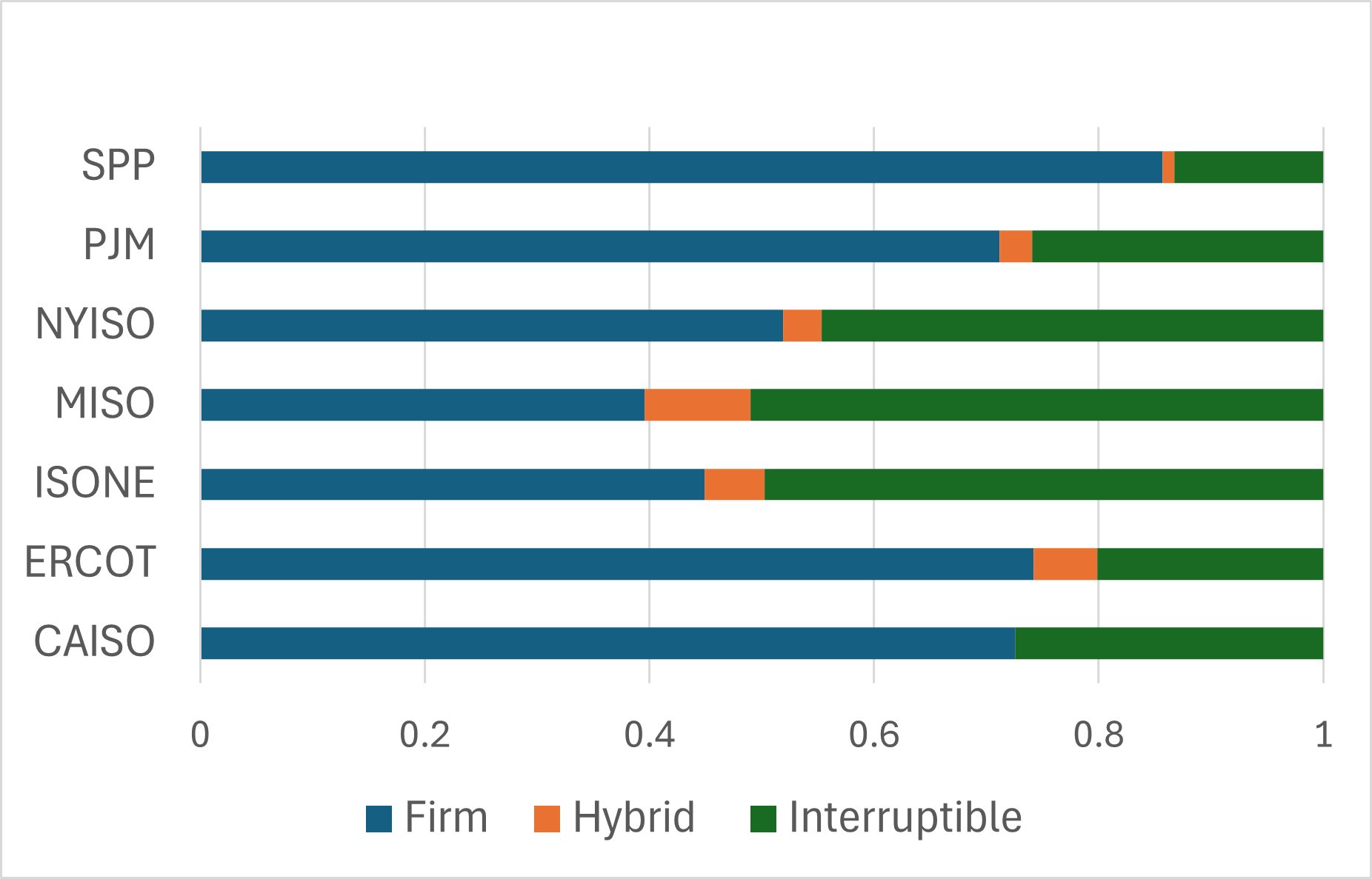}
    \caption{Natural Gas Power Plant Fuel Contracts by Type 2024 }
    \label{fig:enter-label}
\end{figure}

\begin{figure}
    \centering
    \includegraphics[scale = 0.5]{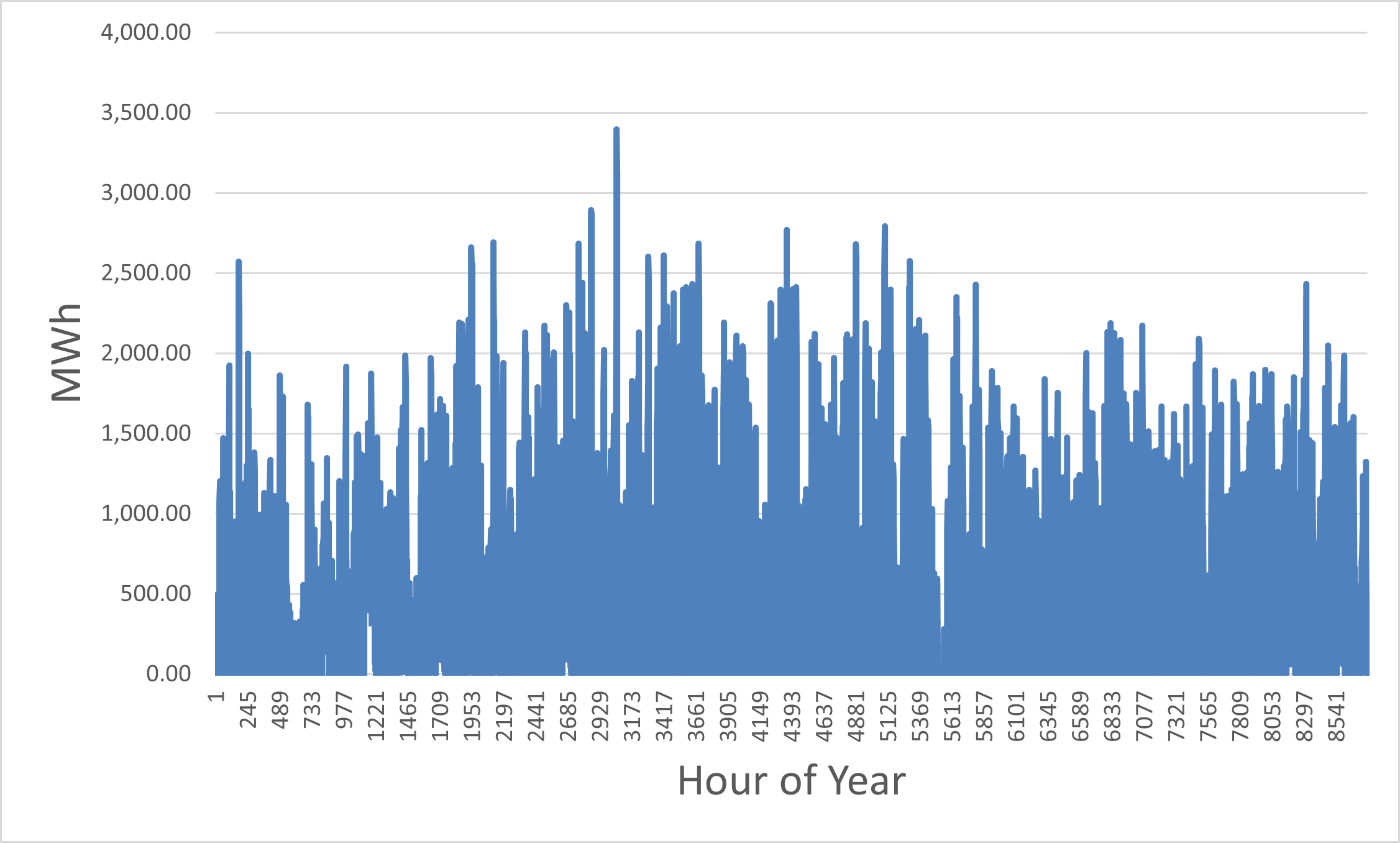}
    \caption{2024 Energy Gap} 
    \label{fig:enter-label}
\end{figure}

In order to address these challenges \cite{white}, ISO NE introduced an enhanced DAM on February 28th, 2025, including a novel ancillary service product called (EIR) \cite{clements_186_nodate}. EIR is designed to balance the gap between the real-time load predicted by ISO NE and energy cleared in the DAM. The day-ahead market will now clear both day-ahead energy and energy imbalance reserve to meet the load forecast. EIR is an energy call "option" product in the DAM where sellers receive an EIR payment in DAM and potentially may incur a close-out cost based on real-time energy prices and a strike price set by the ISO. 

One of ISO NE's central objectives of introducing the EIR mechanism is to better align individual market participant's incentives with the system's overall reliability objectives. The ISO must ensure that sufficient generation is available to meet demand in all possible real-time scenarios. For a natural gas generator, this means it should secure adequate fuel supply ahead of the operating day instead of relying on the spot market in real time. However, purchasing fuel ahead of time exposes generators to the risk of financial loss if they are not dispatched, rendering fuel unused or resold. The goal of the EIR product is to mitigate this risk by providing a market-based mechanism where generators receive a day-ahead payment for EIR awards and face a real-time closeout cost if real-time energy prices exceed a pre-determined strike price. By providing a financial signal ahead of real time, EIR may motivate generators to procure fuel when it is socially desirable but not privately profitable without the EIR mechanism. In addition, by incorporating the ISO's real-time load forecast as the forecast energy requirement (FER) constraint in the DAM and pricing it directly, the EIR product improves DAM price formation and reduces the need for out-of-market actions such as RAA.

\section{Modeling Framework, Literature Review and Contributions}
The EIR product is the first option product of its kind that is cleared in conjunction with other energy and ancillary service products in the wholesale electricity market. The effectiveness of the product in terms of fuel procurement has not been studied in the literature. In this paper, we establish a framework for examining how the EIR option mechanism affects market participants behaviors and fuel procurement incentives. The models capture three types of agents, namely generators, demand and arbitragers, who participate in both the day-ahead and real-time markets under the uncertainty in real-time load and fuel prices. 

The model follows a two-stage market structure, consistent with ISO's practices  \cite{isone_ed}, \cite{hoganpjmed}, \cite{ferc_ed}. In the first stage (day-ahead), participants make decisions based on foretasted real-time conditions. In the second stage (real-time), under each scenario, which represents a possible real-time demand and fuel price realization, only generators make decisions on whether or not to purchase additional fuel and real time production quantities.

    More specifically, we have: 
\begin{itemize}
\item Generators: day-ahead energy awards, EIR awards, and day-ahead fuel purchase.
\item Demand: day-ahead bid-in demand quantity.
\item Arbitragers: day-ahead energy awards
\end{itemize}
In the second stage (real-time), under each scenario, which represents a possible real-time demand and fuel price realization, agents make the following decisions:
\begin{itemize}
\item Generators: purchase additional fuel at the gas spot market if needed, produce real-time energy, and settle the EIR product based on real-time prices and predetermined strike prices. 
\item Demand: settles its realized consumption. 
\item Arbitragers: settle their positions based on real-time prices.
\end{itemize}

It is important to develop the two-stage stochastic structure because market participants must make key decisions, such as day-ahead energy sales or purchases, EIR awards and day-ahead fuel purchase before the realization of uncertain real-time conditions. The stochastic nature of the model captures the uncertainty of possible real-time outcomes and enables us to evaluate risk-hedging behavior with and without the EIR product.

The framework compares two market paradigms:
\begin{itemize}
    \item Energy Market Only (EMO): A baseline energy only model without EIR that is a simplified reflection of ISO-NE's market prior to February 28, 2025
    \item Energy Market with Imbalance Reserve (EMIR): A co-optimized model where the EIR product is procured to meet a system-level forecast energy requirement (FER). 
\end{itemize}

To capture the impact of risk preference on market outcomes and participants' behaviors, the objective functions of participants' profit maximization problem use Conditional Value-at-Risk (CVAR). CVAR represents the generator's expected profits that fall below a percentile of the profit distribution. We adopt the CVAR formulation from Rockafellar and Uryasev \cite{cvar}. 

Our analysis is based on a competitive equilibrium framework. Under this framework, agents are price takers who maximize their profit given market clearing prices. They bid truthfully and do not exercise market power to manipulate market clearing outcomes. The competitive equilibrium model enables us to analyze how risk preferences and market structure influence outcomes without the confounding effect of agents' strategic manipulation.

The resulting market equilibrium is the solution to a mixed complementarity problem (MCP) derived from the optimality conditions of the agents' profit maximization problems and market-clearing constraints. The use of mixed complementarity problems to model the equilibrium of multiple agents participating in an electricity market is well established in the operations research literature on power markets. \cite{dart_eqb} models generators participating in both day-ahead and real time markets.  \cite{gencaprisk} and \cite{smeers_co-optimization_2022} use two-stage stochastic equilibrium models to explore capacity expansion with risk-averse agents and utilize CVAR as their risk measure. \cite{gurkan_generation_nodate} and \cite{de_maere_daertrycke_investment_2017} use MCPs to compare different market structures and policies, specifically how the presence of risk mitigation instruments changes investment behavior and social welfare.  A full treatment of mixed complementarity modeling in energy systems can be found in \cite{gabriel_complementarity_2013} and a review of applications in electricity markets can be found in \cite{dimitriadis_review_2021} and \cite{ruiz_tutorial_2014}. Also, see \cite{murphy_tutorial_2016} for an introduction and tutorial on the use of complementarity problems to analyze different market policies. All the MCPs discussed in this paper were implemented in GAMS and solved using the PATH solver \cite{dirkse_path_1995}.

To maintain tractability and focus on impact of EIR production on equilibrium outcomes, we adopt the following modeling simplifications:
\begin{itemize}
\item Unit commitment and other intertemporal constraints are ignored in the equilibrium models. Also, in reality, fuel procurement can be binary or lumpy decisions. However, modeling binary variables introduces nonconvexity and create computational challenges that are beyond the scope of this paper. Our continuous model approximates fuel decisions at an aggregate level, which still provides meaningful insights into incentive effects of EIR. Generators are assumed to be flexible enough to be dispatched from zero to their maximum capacity, allowing us to use equilibrium models to explore their behaviors in a tractable manner. 
\item Arbitragers are assumed to be purely financial risk-neutral participants who maximize expected profit and face no transaction or administrative costs. In practice, arbitragers may incur small costs that result in a small difference between the day-ahead price and the expected real-time price. However, these costs are typically small and do not significantly affect equilibrium behaviors. This simplification is also a standard approach in equilibrium modeling. It facilitates analytical comparisons across different market designs.

\item The system is modeled as a single-node system so that it helps the analysis focus on the core economic impact of EIR. Ignoring transmission constraints does not affect the main conclusions obtained in this paper. 

\item The traditional operating reserve products, such as 10-minute and 30-minute reserves are not included in the model. The exclusion of these reserve products allows us to focus the model specifically on the EIR product and its intended function of incentivizing advance fuel procurement. While these operating reserve products are important for system reliability, their exclusion does not undermine the ability of the proposed equilibrium model to evaluate the specific role of EIR in influencing market clearing outcomes and fuel procurement behaviors.

\end{itemize}

Our modeling approach focuses on exploring two different aspects. First, we want to understand how the presence of energy imbalance reserve affects the solution of the equilibrium problems, specifically how it changes generators' fuel purchase decision and the behavior of a demand agent bidding in the day-ahead market. Secondly, the two sets of models, risk neutral and risk averse, shed some light in circumstances in which the introduction of EIRs may be an effective policy for achieving its design objectives.

The main contributions and findings of this paper are:
\begin{itemize}
    \item This is the first work to develop a mixed complementarity model to represent the competitive equilibrium of multiple generator and demand agents participating in an energy and reserve market under uncertainty to gain insight into what may occur as a result of the introduction of EIRs.
    \item It is shown that when participants are risk neutral and generators' total capacity exceeds the forecasted real-time demand, the EIR product may not impact market outcomes or participants' fuel procurement behaviors relative to an energy-only DAM in the absence of unit commitment constraints. This finding highlights that participant risk preferences play an important role in realizing the intended benefits of the EIR product.
    \item When participants are risk averse, it is shown that the presence of EIR can change generators' behaviors. In particular, generators become more willing to procure firm fuel in advance of operating day. The EIR mechanism helps shift risk-bearing from the system operator to risk-averse market participants.  
\end{itemize}
The equilibrium modeling framework developed in this paper can be extended to analyze other market design questions related to uncertainty, fuel security, interactions between reliability and financial products, e.g. forward reserves.

\section{Two-Stage Stochastic Equilibrium Models for Markets With and Without EIR}

In this section, we present two-stage stochastic competitive equilibrium models that characterize the behavior of three types of market participants, generators, demand and arbitragers, in both the conventional EMO and newly enhanced EMIR. For each agent type, we develop its profit maximization problem incorporating its risk preferences via CVAR. These optimization problems are formulated separately for the market with energy only as well as the market with energy and EIR.

To obtain competitive equilibrium, we simultaneously solve all agents' optimization problems, along with market clearing conditions. This ensures that no participant has an incentive to deviate and the market reaches an equilibrium. Mathematically, this is achieved by converting each agent's optimization problem into its Karush-Kuhn-Tucker (KKT) conditions. By stacking the KKT conditions for all agents and market clearing conditions, we formulate an MCP. The solutions to this MCP correspond to equilibrium outcomes of the EMO or EMIR.

In a traditional electricity market, supply and demand submit price-quantity pairs that are used to clear the market. In our models, the agents' decisions are restricted to the quantities that they bid or offer into the electricity market. The prices are determined by the full equilibrium model and are set at the level to balance supply and demand across the different markets. While our formulation differs from the actual market clearing process, it is a reasonable approach if we assume that each agent bids in its true cost.  This difference does not affect the validity of the competitive equilibrium results because price-taking behaviors assumes agents respond to market prices rather than influence them, and the clearing prices in our model still reflect marginal cost principle under equilibrium conditions.
    \begin{table}
    \caption{Sets}
    \label{tab:my_label}
    \centering
    \begin{tabular}{cc}
    \hline 
       $i \in \mathcal{I}$ & Generators\\
       $s \in \mathcal{S}$ & Real Time Demand Scenarios \\
       
    \hline
    \end{tabular}
    
     \end{table}

\begin{table}
 \caption{Variables with Index i pertain to Generator i and Variables with Index s pertain to Scenario s}
    \label{tab:my_label}
    \centering
    
    \begin{tabular}{cc}
    \hline 
       $g_i^{DA}$  & day-ahead Energy Award (MWh)\\
       $e_i$ & Energy Imbalance Reserve (MWh) \\
       $g_{i,s}^{RT}$  & Real Time Production (MWh)\\       
       $V^{DA}_{i}$  & day-ahead Fuel Purchase (MWh)\\
       $V^{RT}_{i,s}$  & Spot Fuel Purchase (MWh)\\
       $W^{RT}_{i,s}$ & Remaining Fuel (MWh) \\
         $\lambda_s^{RT}$  & Real Time Energy Price (\$/MWh)\\
         $\lambda^{DA}$  & day-ahead Energy Price (\$/MWh)\\
         $\rho$ & Energy Imbalance Reserve Price\\
        $d^{DA}$  & day-ahead Demand (MWh)\\ 
        $\mu_{i,s}$ & Marginal Value of Fuel Investment (\$/MWh) \\
        $\eta$ & Value at Risk \\
        $q_{i,s}$ & Risk Adjusted Probability \\
        $Z_{i,s}$ & Per Scenario Profit of Generator i\\
        $Z^D_{s}$ & Per Scenario Profit of Demand\\
     \hline
    \end{tabular}
   
\end{table}
\quad
\begin{table}
\caption{Parameters}
    \label{tab:my_label}
    \centering
    \begin{tabular}{cc}
    \hline 
       $C_i$  & Marginal Production Cost Excluding Fuel Costs(\$/MWh)\\\
       $C^I_{i,s}$  & Intraday Fuel Cost (\$/MWh)\\
       $C^F_i$  & Advanced Fuel Cost (\$/MWh)\\
       $R_{i,s}$ & Fuel Resale Price \\
       $F_i$ & Generator Fuel Level (MWh)\\
        $Q_i$ & Maximum Capacity (MWh)\\
        $D^{RT}_s$ & Fixed Real Time Demand (MWh) \\
        $D^{DA}$ &Fixed day-ahead Demand (MWh) \\
         $FER$ & Forecast Energy Requirement (MWh)\\ 
        $\pi_s$ & Probability of Scenario s \\
        $\alpha_i$ & Risk Level \\
        $K$ & Strike Price (\$/MWh) \\
    \hline
    \end{tabular}
    
\end{table}
\subsection{Equilibrium Model for Energy Market Only}
\subsubsection{Generators' Profit Maximization Problem in EMO} 
The individual generator's optimization problem is to maximize their risk adjusted profits. In our set up, any real time production must be supported by having the fuel to generate in real time with the option to purchase fuel in advance, for a fixed price $C^F_i$, or at a spot price $C^I_{i,s}$ that is random and correlated with real time demand scenario $s$. The generator's problem in absence of EIR products is presented below with the associated dual variables in brackets to the right of the constraints.
\begin{subequations}
\begin{align}
\max_{\Xi}&\textcolor{white}{s} \eta_i- \frac{1}{a_i}\sum_s \pi_su_{i,s}   \\
\mbox{s/t}&\textcolor{white}{s}u_{i,s} \ge \eta_i - Z_{i,s}&[q_{i,s}] \\
&g^{DA}_i \le Q_i &[\delta_i] \\
& g^{RT}_{i,s} \le Q_i &[\gamma_{i,s}]\\
& g^{RT}_{i,s} = F_i + V_{i}^{DA}  + V_{i,s}^{RT}- W_{i,s}^{RT} &[\mu_{i,s}]\\
& g^{RT}_{i,s},g^{DA}_i,V^{DA}_i,V^{RT}_{i,s},W^{RT}_{i,s},u_{i,s} \ge 0 &  \forall s 
\end{align}
\end{subequations}

where
\begin{align}
\Xi := (g^{RT}_{i,s},g^{DA}_i,u_{i,s},V^{DA}_i,V^{RT}_{i,s},W^{RT}_{i,s},\eta_i,u_{i,s})\nonumber\\
  Z_{i,s} := \lambda^{DA}g^{DA}_i +\lambda^{RT}_s(g^{RT}_{i,s}-g^{DA}_i)-\label{Gen_CVARprofit_EMO} \\ C_ig^{RT}_{i,s}- C^{F}_iV_{i}^{DA} - C^{I}_{i,s}V_{i,s}^{RT} + R_{i,s}W_{i,s}^{RT} \nonumber
 \end{align}
The generator's objective function captures the CVAR of its total day-ahead and real-time profit across uncertain real-time scenarios. This formulation follows the approach introduced by Rockafellar and Uryasev \cite{cvar}. In the model, the generator maximizes its risk adjusted profit by choosing a benchmark profit level $\eta_i$ and penalizing the expected profit below that level using $u_{i,s}$. The non-negative variable $u_{i,s}$ reflects how much the actual profit in scenario $s$ is below this benchmark profit level $\eta_i$ via the first inequality constraint. In other words, The generator agent $i$ hedges against its lowest $\alpha_i$ quantile of profits. The parameter $\alpha_i\in(0,1]$ corresponds to the generator's level of risk aversion.   When $\alpha_i=1$, the generator behaves in a risk-neutral manner and maximizes its expected profit. As $\alpha_i$ decreases, the generator becomes more risk-averse, emphasizing performance in the tail of the profit distribution. This CVAR formulation allows us to examine generator's behaviors under different risk preferences while keeping the model tractable.

The second and third inequality constraints are related to physical limitations on real-time and day-ahead decision variables, where $Q_i$ represents the economic maximum or capacity of a generator. The fourth equality constraint enforces $g_{i,s}^{RT}$ to be equal to some standard fuel level $F_i$, which can be thought of as a generator's starting fuel position before investing in additional fuel on the gas market and in this paper is set at a level of zero for all of the numerical examples, plus any day-ahead $V^{DA}_i$ or real time $V^{RT}_{i,s}$ fuel investment, and less any fuel leftover $W^{RT}_{i,s}$, which can be thought of as fuel to be sold after the real-time market is closed. The corresponding shadow prices, $\mu_{i,s}$, can be interpreted as the marginal value of fuel investment.

\subsubsection{Demand's Profit Maximization Problem in EMO}\label{EMO_demand_profit_section} 
In this section, we provide a demand agent that represents the aggregate of all demand side market participants and pays for all energy in EMO. $d^{DA}$ is the sole decision variable of this agent, as we assume their real time loads are fixed and represented by $D^{RT}_s$. Demand agent's problem in absence of EIR options can be described as follows:
\begin{subequations}
\begin{align}
\underset{{\eta^D,u_s^D,d^{DA}}}{\mbox{max}} & \eta^D - \frac{1}{\alpha^D}\sum{\pi_s}u_{s}^D  \\
\mbox{s/t}& u_{s}^D \ge \eta^D - Z_{s}^D& [q_{s}^D] \\
& u_{s}^D,d^{DA} \ge 0 & \forall s
\end{align}
\end{subequations}
where
\begin{equation}
  Z_{s}^D := -\lambda^{DA}d^{DA}-\lambda^{RT}_s(D^{RT}_{s}-d^{DA}) \label{dem_CVARprofit_EMO} \\
\end{equation}
The demand agent maximizes its risk adjusted profit, similar to the generator's counterpart. When $\alpha^D=1$, the demand agent is risk neutral.

\subsubsection{Arbitragers' Profit Maximization Problem in EMO} \label{EMO_arbitragers_profit_section} 

 Every ISO market in the US allows for the participation of virtual market participants. These entities are strictly financial in that they neither produce or consume electricity in real time, but merely seek to exploit arbitrage opportunities between the day-ahead and real time markets. They are able to offer in virtual supply (commonly referred to as incremental offers or INCs) or virtual demand (decremental bids or DECs), in the day-ahead market which is then closed out in real time. The arbitragers' problem is to maximize expected profits through their decision variables, $a$ and $b$, which respectively represent INCs and DECs.

\begin{subequations}
\begin{align}
\underset{a,b}{\mbox{max}}& \textcolor{white}{hhh} E_s[Z_{s}] \\
\text{s/t} &\textcolor{white}{hhh}a,b \ge 0\\
  &Z_{s} := (\lambda^{DA} - \lambda^{RT}_s)a +  (\lambda^{RT}_s - \lambda^{DA})b
\end{align}
\end{subequations}

Arbitragers are assumed to be always risk-neutral. Since their positions are purely financial, this type of agent is less exposed to operational risk and thus behaves rationally by maximizing expected profit without penalizing rare unfavorable outcomes. In practice, arbitragers may face administrative or transaction costs that make them unwilling to act on marginal arbitrage opportunities unless the expected profit exceeds a small thresholds, typically around \$1-\$2/MWh. As a result, real-world arbitrage activity may not perfectly equalize day-ahead energy prices and expected real-time prices. However, for the purpose of this model, we assume costless arbitrage. This simplification allows us to clearly identify the role of EIR and risk preference in impacting equilibrium outcomes. Moreover, since the transaction costs are small relative to the scale of market clearing prices and its movements induced by participants risk preference or market structure changes, this assumption does not materially affect the qualitative insights obtained from the equilibrium analysis.

\subsubsection{Market Clearing Conditions in EMO}
The DA energy price is determined by the
complementarity condition
\begin{equation}
\sum_i g^{DA}_i + a = d^{DA} + b  \perp \lambda^{DA} \mbox{ free} \label{DAenergyMC}
 \end{equation}
by ensuring total day-ahead energy supply and INCs equals total day-ahead bid-in demand and DECs. $\lambda^{DA}$ represents the marginal cost of supplying one more unit of energy in the day-ahead market. Because the condition is an equality, the associated price $\lambda^{DA}$  is unrestricted in sign.
The RT energy price is determined by the
complementarity condition
\begin{equation}
\sum_i g^{RT}_{i,s} = D^{RT}_s  \perp \lambda^{RT}_s \mbox{ free} \label{RTenergyMC}
\end{equation}
by ensuring total real-time energy production equals real-time materialized demand at each scenario. $\lambda^{RT}_s$ represents the marginal cost of supplying one more unit of energy in the real-time market under scenario $s$.
Note that transmission constraints are ignored in the market clearing conditions for the simplification purpose. As a result, all agents face the same energy prices. This assumption is not expected to affect the qualitative findings of the paper. In practice, the New England system is not frequently congested.     
\subsubsection{Full MCP for EMO} 
The full MCP for the EMO equilibrium is formed by appending the generators' optimality conditions to those of the arbitragers and demand agent, along with market clearing conditions requiring that supply equal demand in the real time and day-ahead markets. This is presented in full below:
\begin{align}
0 \le C_iq_{i,s} + \mu_{i,s} - \lambda^{RT}_sq_{i,s} + \gamma_{i,s} & \perp g^{RT}_{i,s} \geq 0 \label{RT1cvar}\\
0 \le C^{F}_i  - \sum_s\mu_{i,s} & \perp V_{i}^{DA} \ge 0 \label{DA1cvar} \\
0 \le C^{I}_{i,s}q_{i,s} - \mu_{i,s} & \perp V_{i,s}^{RT} \ge 0 \label{RTFcvar}\\
0\leq \mu_{i,s} - R_{i,s}^{RT}q_{i,s} &\perp W^{RT}_{i,s} \geq 0 \label{RT_fuel_sold_EMO}\\
0 \le \sum_s q_{i,s}\lambda^{RT}_{s} - \lambda^{DA} + \delta_i & \perp g^{DA}_i \ge 0 \label{DA2cvar}\\ 
\sum_s q_{i,s}  = 1 &\perp \eta_i \mbox{ free} \label{emocvar2} \\
   0 \le\frac{1}{\alpha_i}\pi_s - q_{i,s}& \perp u_{i,s} \ge 0 \label{emocvardem3}\\
0 \le Q_i - g^{RT}_{i,s} & \perp \gamma_{i,s} \ge 0 \label{RT2cvar} \\
g^{RT}_{i,s} = F_i + V_{i,s}^{RT}+ V_{i}^{DA} - W_{i,s}^{RT} & \perp  \mu_{i,s} \mbox{ free} \label{RT3cvar}\\
0 \le Q_i - g^{DA}_i & \perp \delta_i \ge 0 \label{DACapcvar}\\
  0 \le u_{i,s} - \eta_{i} + Z_{i,s} & \perp q_{i,s} \ge 0 \label{emocvar1}\\
   0 \le \lambda^{DA} - \sum_s q_{s}^D\lambda^{RT}_s & \perp d^{DA}\ge 0  \label{DDA1dem_EMO} \\
    0 \le u_{s}^D - \eta^D + Z_{s}^D & \perp q_{s}^D \ge 0\label{Demandcvar1}\\
   \sum_s q_{s}^D  = 1 &\perp \eta^D \mbox{ free} \label{Demandcvar2} \\
   0 \le \frac{1}{\alpha^D}\pi_s - q_{s}^D& \perp u_{s}^D \ge 0 \label{demandcvardem3}\\
0 \le E_s[\lambda^{RT}_s] - \lambda^{DA} & \perp a \ge 0 \label{ARBT1cvar} \\
0 \le \lambda^{DA} - E_s[\lambda^{RT}_s] & \perp b \ge 0 \label{ARBT2cvar} \\
\sum_i g^{DA}_i + a = d^{DA} + b & \perp \lambda^{DA} \mbox{ free} \label{MC1cvar}\\
\sum_i g^{RT}_{i,s} = D^{RT}_s  & \perp \lambda^{RT}_s \mbox{ free} \label{MC2cvar}
\end{align}

where $Z_{i,s}$ and $Z_s^D$ are defined by (\ref{Gen_CVARprofit_EMO}) and (\ref{dem_CVARprofit_EMO}), respectively.

Conditions (\ref{RT1cvar})-(\ref{emocvar1}) correspond to the KKT optimality conditions of generators' profit maximization problem.

Conditions (\ref{DA1cvar}) and (\ref{DA2cvar}) are originally formulated as nonlinear complementarity conditions, derived from the KKT conditions of the generators' profit maximization problem:

\begin{align}
0 \le C^{F}_i\sum_s q_{i,s}  - \sum_s \mu_{i,s}  &\perp V_{i}^{DA} \ge 0\\
0 \le \sum_s q_{i,s}\lambda^{RT}_{s} - \lambda^{DA}{\sum_s q_{i,s}}  + \delta_i &\perp g^{DA}_i \ge 0 \
\end{align}

Because of condition (\ref{emocvar2}), these nonlinear complementarity conditions can be simplified and equivalently expressed as 
conditions (\ref{DA1cvar}) and (\ref{DA2cvar}).

In addition, the dual variable, $q_{i,s}$, can be interpreted as the risk-adjusted probability that generator $i$ places on scenario $s$. When the generator is risk neutral, these probabilities are equal to $\pi_s$ due to condition (\ref{emocvardem3}).

Conditions (\ref{DDA1dem_EMO})-(\ref{demandcvardem3}) represent the KKT optimality conditions of demand agent's profit maximization problem. Condition (\ref{DDA1dem_EMO}) was initially expressed as a nonlinear complementarity condition, derived from the KKT conditions of the demand agent's profit maximization problem:

\begin{align}
 0 \le \lambda^{DA}{\sum_s q_{s}^D} - \sum_s q_{s}^D\lambda^{RT}_s & \perp d^{DA}\ge 0  
\end{align}

Given condition (\ref{Demandcvar2}), this nonlinear complementarity condition can be equivalently formulated in a simplified form as condition (\ref{DDA1dem_EMO}).

Conditions (\ref{ARBT1cvar}) and (\ref{ARBT2cvar}) are arbitragers profit maximization problem's KKT optimality conditions. Conditions (\ref{MC1cvar}) and (\ref{MC2cvar}) are market clearing conditions of EMO.

\subsection{Equilibrium Model for Market with Imbalance Reserve}
\subsubsection{Generators' Profit Maximization Problem in EMIR} 
 We similarly model the generator $i$'s optimization problem when EIR options are present.
\begin{subequations}
\begin{align}
\max_{\Xi} &\textcolor{white}{h} \eta_i- \frac{1}{a_i}\sum_s \pi_su_{i,s}\\
\mbox{s/t}& \textcolor{white}{h}u_{i,s} \ge \eta_i- Z_{i,s}&[q_{i,s}] \\
&g^{DA}_i + e_i \le Q_i &[\delta_i] \\
& g^{RT}_{i,s} \le Q_i &[\gamma_{i,s}]\\
& g^{RT}_{i,s} = F_i + V_{i}^{DA} + V^{RT}_{i,s} - W_{i,s}^{RT} &[\mu_{i,s}]\\
& g^{RT}_{i,s},g^{DA}_i,e_i,V^{DA}_i,V^{RT}_{i,s},W^{RT}_{i,s},u_{i,s} \ge 0 &\forall s 
\end{align}
\end{subequations}

where
\begin{align}
\Xi :=(g^{RT}_{i,s},g^{DA}_i,u_{i,s},V^{DA}_i,V^{RT}_{i,s},W^{RT}_{i,s},e_i,\eta_i) \nonumber\\
  Z_{i,s} := (\lambda^{DA}+\rho)g^{DA}_i +\lambda^{RT}_s(g^{RT}_{i,s}-g^{DA}_i)+\rho e_i \nonumber \\ 
  -[\lambda^{RT}_s - K]^{+}e_i - C_ig^{RT}_{i,s} - C^{F}_iV_{i}^{DA} - C^{I}_{i,s}V_{i,s}^{RT}+R_{i,s}W_{i,s}^{RT}. \label{Gen_CVARprofit_EMIR} 
\end{align}

 Generators offer in EIRs at a price and quantity, much like how day-ahead energy offers are made (EIR being a financial product, we assume the generator does not incur any additional costs when providing EIR). In the day-ahead EMIR, a generator earns the EIR clearing price, $\rho$, for each MWh of reserve they provided, represented by $e_i$. The EIRs are purchased by the system at a strike price of $K$, which is determined (in advance and announced) by the system operator. In real-time any seller of EIR options, is charged $[\lambda_s^{RT} -K]^+$ for each MWh of EIR they have sold in the day-ahead. Note that the symbol + indicates there will be no cost to the generator provided $K$ is more than the real-time price. This real-time charge on the price of energy is rather unique, as reserves are commonly settled using reserve prices only. Using $\lambda_s^{RT}$ in settling the EIR option implies that higher production in scenario 
$s$ results in lower closeout charges for reserves. In our risk-averse equilibrium models, this effect is manifested in the generator's behavior - specifically, their decision to invest in advance fuel as a hedge against potential closeout charges. See Section V.A for details. Under the EMIR mechanism, the day-ahead energy award $g^{DA}_i$ is not only compensated at the day-ahead energy clearing price $\lambda^{DA}$, but also at $\rho$ for its contribution to stratifying the FER as shown in the market clearing condition (\ref{FER_clearing}).  
\subsubsection{Demand's Profit Maximization Problem in EMIR} 
In this section, we provide a demand agent that represents the aggregate of all demand side market participants and pays for all energy and reserve. $d^{DA}$ is the sole decision variable of this agent, as we assume their real time loads are fixed. Demand agent's problem in the presence of EIR options follows.

\begin{subequations}
\begin{align}
&\underset{{\eta^D,u_s^D,d^{DA}}}{\mbox{max}} & \eta^D - \frac{1}{\alpha^D}\sum{\pi_s}u_{s}^D & \\
&\mbox{s/t}& u_{s}^D \ge \eta^D - Z_{s}^D&[q_{s}^D]& \\
&& u_{s}^D,d^{DA} \ge 0 
\end{align}
\end{subequations}

where
\begin{align}
 Z_{s}^D := -(\lambda^{DA})d^{DA} -\lambda^{RT}_s(D^{RT}_{s}-d^{DA}) \nonumber\\ 
  -\rho(FER)
+ [\lambda^{RT}_s - K]^{+}\sum_i e_i  \label{dem_CVARprofit_EMIR2} 
\end{align}

\subsubsection{Arbitragers' Profit Maximization Problem in EMIR} 
The arbitragers' profit maximization problem in EMIR is identical to the EMO counterpart described in subsection \ref{EMO_arbitragers_profit_section}.

\subsubsection{Market Clearing Conditions in EMIR}
The DA and RT energy price complementarity conditions (\ref{DAenergyMC}) and (\ref{RTenergyMC}) are needed in EMIR. In addition, price for EIR, $\rho$ is determined by the complementarity condition
\begin{align}
0\leq \sum_ig_i^{DA} + \sum_i e_i -FER  \perp \rho \geq 0 \label{FER_clearing}
\end{align}
which stipulates that the EIR clearing price is positive only when FER equals the DA energy awards and EIR awards.

\subsubsection{MCP for EMIR} 
The full MCP for the EMIR equilibrium is formed by the optimality conditions of the generator, demand, arbitragers, and market clearing conditions. 

\begin{align}
0 \le C_iq_{i,s} + \mu_{i,s} - \lambda^{RT}_sq_{i,s}  + \gamma_{i,s} & \perp g^{RT}_{i,s} \geq 0 \label{EMIR_RT1cvar}\\
0 \le C^{F}_i  - \sum_s \mu_{i,s} & \perp V_{i}^{DA} \ge 0 \label{EMIR_DA1cvar} \\
0 \le C^{I}_{i,s}q_{i,s}- \mu_{i,s} & \perp V_{i,s}^{RT}  \ge 0 \label{EMIR_RTFcvar}\\
0\leq \mu_{i,s} - R_{i,s}q_{i,s} &\perp W^{RT}_{i,s} \geq 0 \label{RT_fuel_sold_EMIR}\\
0 \le \sum_s q_{i,s}\lambda^{RT}_{s} - \lambda^{DA} -\rho + \delta_i & \perp g^{DA}_i \ge 0 \label{EMIR_DA2cvar}\\ 
  0\leq \sum_s q_{i,s}[\lambda^{RT}_s - K]^+ -\rho +\delta_i &\perp e_i \geq 0 \label{EMIR_EIRdualcvar}\\
  \sum_s q_{i,s}  = 1 &\perp \eta_i \mbox{ free} \label{EMIR_emocvar2} \\
     0 \le \frac{1}{\alpha_i}\pi_s - q_{i,s}& \perp u_{i,s} \ge 0 \label{EMIR_emocvardem3}\\
0 \le Q_i - g^{RT}_{i,s} & \perp \gamma_{i,s} \ge 0 \label{EMIR_RT2cvar} \\
g^{RT}_{i,s} = F_i + V_{i,s}^{RT}+ V_{i}^{DA} - W_{i,s}^{RT} & \perp  \mu_{i,s} \mbox{ free} \label{EMIR_RT3cvar}\\
0 \le Q_i - g^{DA}_i -e_i& \perp \delta_i \ge 0 \label{EMIR_DACapcvar}\\
  0 \le u_{i,s} - \eta_{i} + Z_{i,s} & \perp q_{i,s} \ge 0 \label{EMIR_emocvar1}\\
     0 \le \lambda^{DA} - \sum_s q_{s}^D\lambda^{RT}_s & \perp d^{DA}\ge 0 \label{EMIR_DDA1dem_alterantive} \\
   0 \le u_{s}^D - \eta^D + Z_{s}^D & \perp q_{s}^D \ge 0\label{EMIR_Demandcvar1}\\
   \sum_s q_{s}^D  = 1 &\perp \eta^D \mbox{ free} \label{EMIR_Demandcvar2} \\
   0 \le \frac{1}{\alpha^D}\pi_s - q_{s}^D& \perp u_{s}^D \ge 0 \label{EMIR_demandcvardem3}\\
0 \le E_s[\lambda^{RT}_s] - \lambda^{DA} & \perp a \ge 0 \label{EMIR_ARBT1cvar} \\
0 \le \lambda^{DA} - E_s[\lambda^{RT}_s] & \perp b \ge 0 \label{EMIR_ARBT2cvar} \\
\sum_i g^{DA}_i + a = d^{DA} + b & \perp \lambda^{DA} \mbox{ free} \label{EMIR_MC1cvar}\\
\sum_i g^{RT}_{i,s} = D^{RT}_s  & \perp \lambda^{RT}_s \mbox{ free} \label{EMIR_MC2cvar}\\
0\leq \sum_ig_i^{DA} + \sum_i e_i -FER  &\perp \rho \geq 0 \label{EIR_MCFER}
\end{align}
where $Z_{i,s}$ and $Z^D_s$ are defined by (\ref{Gen_CVARprofit_EMIR}) and (\ref{dem_CVARprofit_EMIR2}).

\section{Theoretical Results of the Risk Neutral Case}
In the following section, we state theoretical findings for the risk neutral mixed complementarity models. These models are special cases of the risk models when $\alpha_i = 1$, which represents agents maximizing their expected profits. 
Propositions \ref{prop:1} and \ref{prop:2}, Corollary \ref{cor:1}, as well as Lemmas \ref{lemma:AdvFuel1} and 
\ref{lemma:AdvFuel2} refer to the risk neutral model with a demand agent, shown in the Appendix A1 and A2. 


\begin{proposition}
\label{prop:1}
   Suppose that total available capacity is greater than $FER$ (i.e. $\sum_i Q_i > FER$). Then any valid solution to the model with energy imbalance reserve and virtual traders, satisfies $\rho = 0$. Furthermore, when the close out cost is not identically zero across all scenarios, $e_i =0$ in equilibrium.
\end{proposition}
\begin{proof}
    Observe from (\ref{eirARBT1dem}) and (\ref{eirARBT2dem}) that the presence of arbitragers results in $\lambda^{DA} =  E_s[\lambda^{RT}_s]$. Furthermore, since $\sum_i Q_i > FER$, either $\sum_ig_i^{DA} + \sum_i e_i > FER$ in which case $\rho =0$ and we are done, or it must be that for some $\hat{i}$, 
    $g_{\hat{i}}^{DA} + e_{\hat{i}} < Q_{\hat{i}}$. Therefore, from (\ref{eirDA2dem}), we must have $\delta_{\hat{i}} = 0$. 
    Now, (\ref{eirDA1dem}) (together with the arbitrage condition) implies that $\rho \leq \delta_{i}$ for any $i$, and in particular, considering $\hat{i}$, this yields $\rho \leq 0$. Together with the right-hand-side of (\ref{eirMC3dem}) this proves the result. 

    We make a further observation that if $E([\lambda_s^{RT}-K]^+) >0$, that is, if the close out cost of the option is not identically zero in {\em{all}} scenarios, $e_i =0$. This is follows directly from (\ref{eirEIRdem}) together with $\delta_i \geq 0$ from (\ref{eirDA2dem}) and the fact that $\rho$ must be zero as proved above. 
\end{proof}

\begin{cor}
\label{cor:1}
    Under the same condition as Proposition \ref{prop:1}, any equilibrium of the market with EIRs can be transformed to an equilibrium of the market without EIR options where the day-ahead fuel purchases, spot fuel purchases, and the total day-ahead procurement remains the same. 
\end{cor}
\begin{proof}
Note that in the MCP systems, conditions (\ref{RT1dem}--\ref{DA2dem}) are identical to (\ref{eirRT1dem}--\ref{eirDA2dem}), conditions (\ref{ARBT1dem}--\ref{MC2dem}) are identical to (\ref{eirARBT1dem}--\ref{eirMC2dem}), and lastly (\ref{DDA1dem}) is identical to (\ref{eirDDA1dem}). Any difference in the equilibrium systems of the two market schemes is attributed to the existence of $e_i$ and $\rho$. Now, let us denote the solutions to an equilibrium of the market with EIRs by overline. Given a solution to the market with EIRs, construct a solution to the market without EIRs as follows:
\begin{itemize}
    \item Keep all real time values the same, e.g. $g_{i,s}^{RT} = \overline{g_{i,s}^{RT}}$, $ V_{i,s}^{RT} = \overline{V_{i,s}^{RT}}$, $\lambda_s^{RT} = \overline{\lambda_s^{RT}}$, etc. 
    \item Let $V_i^{DA} = \overline{V_i^{DA}}$.
    \item For the remaining decision variables:
\begin{align*}
g_i^{DA} & = e_i + \overline{g_i^{DA}}, \\
b & = \sum_i e_i + \overline{b},\\
\lambda^{DA} & = \overline{\lambda^{DA}}, \\
\delta_i & = \overline{\delta_i}.
\end{align*}
\end{itemize}
It is easy to see the above provides an equilibrium solution satisfying (\ref{RT1dem}--\ref{DDA1dem}). This completes the proof. 
\end{proof}

{\bf{Remark}}: {\textcolor{black}{If $g_i^{DA}$ in the equilibrium solution of EMO satisfy $\sum_i g_i^{DA} \geq FER$, we can construct a solution for the market with EIRs by the following simple transformation: }}
\begin{align*}
\overline{g_{i,s}^{RT}} & = g_{i,s}^{RT},\\ 
\overline{\lambda_s^{RT}} & = \lambda_s^{RT}\\
\overline{g_i^{DA}} & = g_i^{DA}, \\
\overline{\lambda^{DA}} & = \lambda^{DA}, \\
\rho & = 0,\\
e_i & = 0.
\end{align*}
{\textcolor{black}{However, if $\sum_i g_i^{DA} < FER$ the natural transformation above will not render an equilibrium for EMIR. Note that any equilibrium solution to the EMIR will guarantee that $\sum_i g_i^{DA} \geq FER$ by its nature, as per Proposition \ref{prop:1}.}}

\begin{proposition}
\label{prop:2}
    In each of the EMIR and EMO markets described above, as long as $\lambda_s^{RT}$ are unique, we obtain unique equilibrium prices. The day-ahead and energy imbalance reserve quantities for an equilibrium solution are however not unique.
\end{proposition}
\begin{proof}
    Note that the arbitrage conditions enforce $\lambda^{DA} = E_s[\lambda_s^{RT}]$. As long as we have unique real time prices, all equilibria will exhibit uniqueness of day-ahead (and clearly real time) prices. However, the non-uniqueness of the equilibria corresponds to the (infinite) continuum of eligible $b$ and $g_i^{DA}$ that satisfy $\sum_i g_i^{DA} + a = d^{DA} + b$. 
\end{proof}

\begin{lemma}
\label{lemma:AdvFuel1}
    If the expected intraday fuel cost is less than the advanced fuel cost for some generator $i$, i.e. $\sum_s C_{i,s}^I \pi_s < C_i^F$, then $V_i^{DA} = 0$ under both EMO and EMIR market structures, in equilibrium. 
\end{lemma}
\begin{proof}
    Assume on the contrary that $V_i^{DA} >0$. Then from (\ref{DA1dem}) (or equivalently from 
    (\ref{eirDAFdem}) in the EMIR model), we obtain that $C_i^F = \sum_s \mu_{i,s}$. On the other hand (\ref{RTFdem}) (or equivalently ( \ref{eirRTFdem}) in the EMIR model), implies that $\sum_s C_{i,s}^I \pi_s \geq \sum_s \mu_{i,s}$. Putting the above together, we obtain $\sum_s C_{i,s}^I \pi_s \geq C_i^F$, which contradicts the assumption. Hence we have proved $V_i^{DA} = 0$.
\end{proof}

The next lemma highlights that unless the EMIR mechanism results in ``high enough'' real-time energy prices $\lambda_{\hat{s}}^{RT}$, for the scenarios $\hat{s}$ in which a high marginal cost of production generator $i$ is dispatched, then no additional incentives for advance fuel procurement is provided. 
\begin{lemma}
\label{lemma:AdvFuel2}
    Let generator $i$ have a marginal cost of production such that it is dispatched in scenarios 
    $\hat{s} \in \hat{S}$ but not in scenarios ${s'} \in S'$ where $S = \hat{S} \cup {S'}$. If $\sum_s C_{i,s}^I \pi_s \geq C_i^F > \sum_{s'} R_{i, {s'}} \pi_{s'} + \sum_{\hat{s}} (-C_i + \lambda_{\hat{s}}^{RT}) \pi_{\hat{s}}$ then $V_i^{DA} = 0$.
\end{lemma}
\begin{proof}
    Assume on the contrary that $V_i^{DA} >0$. Then from from (\ref{DA1dem}) (or equivalently from 
    \ref{eirDAFdem} in the EMIR model), we obtain that $C_i^F = \sum_s \mu_{i,s}$. For $\hat{s} \in \hat{S}$, we have that $g_{i, \hat{s}}^{RT} >0$, hence from (\ref{RT1dem}) (or equivalently from (\ref{eirRT1dem}) for the EMIR model), we have that 
    $$\mu_{i, \hat{s}} = (\lambda_{\hat{s}}^{RT} - C_i) \pi_{\hat{s}} - \gamma_{i, \hat{s}} \leq (\lambda_{\hat{s}}^{RT} - C_i) \pi_{\hat{s}}. $$
    For $s' \in S'$, we know that $V_i^{DA} >0$, $V_{i, \hat{s}}^{RT} \geq 0$, $F_i \geq 0$, and 
    $g_{i, s'}^{RT} =0$, hence it must be that $W_{i, s'}^{RT} >0$. Therefore (\ref{wconstrdem}) (or equivalently (\ref{eirwconstrdem}) for the EMIR case) implies that 
    $$\mu_{i, s'} = R_{i, s'} \pi_{s'}.$$
    Putting the above together, we obtain 
    $$
    \begin{array}{cl}
    \sum_s \mu_{i,s} & = \sum_{s'} \mu_{i,s'} + \sum_{\hat{s}} \mu_{i,\hat{s}} \\
                     & \leq \sum_{s'} R_{i, s'} \pi_{s'} + \sum_{\hat{s}} (-C_i + \lambda_{\hat{s}}^{RT}) \pi_{\hat{s}},
    \end{array}
    $$
    or equivalently,
    $$C_i^F\leq \sum_{s'} R_{i, s'} \pi_{s'} + \sum_{\hat{s}} (-C_i + \lambda_{\hat{s}}^{RT}) \pi_{\hat{s}}.$$
    This contradicts the assumption, hence we have proved the lemma.
\end{proof}

 Before we proceed to the next sections, we make the observation that in the risk neutral equilibrium models, the real time price in scenario $s$ is given by the marginal cost of production of the marginal generator, plus the fuel investment cost if that generator has acquired any fuel. To this end, observe from equation (\ref{RT2dem}) that $\gamma_{i(s), s} = 0$ in scenario $s$ for the marginal generator 
$i(s)$ in that scenario. Furthermore, if $V_{i(s), s}^{RT} > 0$, from equation (\ref{RTFdem}), we must have 
$\mu_{(i(s), s)} = C_{i(s), s}^I \pi_s$. These together with equation (\ref{RT1dem}) yield that 
$\lambda_s^{RT} = C_{i(s)} + C_{i(s), s}^I$. The same result holds true for the EMIR model as the real time conditions are identical.

\section{Observations and Numerical Examples}
This section presents numerical simulation results that highlight the differences between the EMO and EMIR models. We begin with a simple case involving a single generator and two real-time scenarios to illustrate the underlying reasons for the models’ distinct outcomes. We then extend the analysis to a more complex example with two generators and five real-time scenarios.
\subsection{Single Generator Example}
Consider a single generator, two equally likely real time scenarios, FER = 90 MWh with all remaining input parameters displayed in Table \ref{tab:onegeninputs}. To isolate and understand the fundamental effect of EMIR mechanism in this stylized example, the demand side agent does not participate in the day-ahead market. It does not schedule demand in advance. Instead, the load materializes only in real time. This setup allows us to focus solely on how EMIR influences the generator's incentives for advanced fuel procurement and energy production in the presence of uncertainty.

\begin{table}
\caption{One Generator Example Inputs}
    \label{tab:onegeninputs}
    \centering
    \begin{tabular}{|c|c|c|c|c|l|} \hline 
         $C_1$& $C^I_{1,1}$ & $C^I_{1,2}$ &  $C^{F}_1$& $D^{RT}_1$ &$D^{RT}_2$\\ \hline 
         0&  10&  15&  13&  75&125\\ \hline
    \end{tabular}
    
\end{table}

Table \ref{tab:singlegenemo} shows the outputs of generator's decision variables (fuel investment, day-ahead and EIR awards) in the EMIR model, as well as its profits and risk-adjusted probabilities when FER= 90 MWh and K = \$12. The risk neutral generator opts not to invest in advanced fuel because the cost of the investment,\$13, exceeds its expected marginal benefit of investment, $E[\mu_{i,s}]=\$12.5$. However, the situation changes when the generator is modeled as risk-averse. In this case, decision-making shifts from relying on the expected marginal benefit to the risk-adjusted marginal benefit, represented by $\sum_s q_{i,s}\mu_{i,s}$. As risk aversion increases, the generator becomes more willing to invest in advanced fuel, placing greater weight on adverse scenarios. Specifically, the generator assigns higher risk-adjusted probabilities to the second scenario, where intraday fuel, real time deviation, and reserve closeout costs are elevated.

This shift in behavior leads to a convergence of profits across both scenarios toward zero, stemming from a strong aversion to downside risk. The generator is effectively sacrificing potential upside in more favorable scenarios to eliminate the possibility of incurring losses. Advanced fuel investment acts as a risk mitigation strategy, particularly in the second scenario, where it significantly reduces the impact of incurring high intraday costs and penalties for deviations from day-ahead commitments. Although this investment slightly diminishes profitability in the more favorable first scenario, it offers substantial protection in the second, thereby aligning with the generator’s risk-averse objective to stabilize outcomes across uncertain futures. 
\begin{table}
    \caption{Single Generator EMIR Model Outputs}
    \label{tab:singlegenemo}
    \centering
    \begin{tabular}{|c|c|c|c|c|c|c|c|} \hline 
        $\alpha_1$ & $V_1^{DA}$  &$g^{DA}_1$ &$e_1$& $Z_{1,1}$ & $Z_{1,2}$ & $q_{1,1}$ &$q_{1,2}$\\ \hline 
         1&  0 &90 &0&  263&  -263&  0.5&0.5\\ \hline 
         0.7&  75 &90 &0&  75&  -30&  0.28&0.72\\ \hline 
         0.4&  75 &66.97 &23.03&  0&  0&  0.23&0.77\\ \hline
    \end{tabular}
    
\end{table}

Table \ref{tab:emo1gen} shows results in a market with energy only. The generator does not invest in any advance fuel regardless of its risk level.  There is little difference in the behavior of the generator across different risk levels and it chooses to not participate in the day-ahead market or invest in advance fuel. By purchasing intraday fuel, it can recover its scenario dependent fuel costs through $\lambda^{RT}_s$, resulting in zero profits at every risk level. 

\begin{table}
    \caption{Single Generator EMO Model Outputs}
    \label{tab:emo1gen}
    \centering
    \begin{tabular}{|c|c|c|c|c|c|c|} \hline 
        $\alpha_1$ & $V_1^{DA}$  &$g^{DA}_1$ & $Z_{1,1}$ & $Z_{1,2}$ & $q_{1,1}$ &$q_{1,2}$\\ \hline 
         1&  0 &0 &  0&  0&  0.5&0.5\\ \hline 
         0.7&  0  &0&  0&  0&  0.4&0.6\\ \hline 
         0.4&  0  &0&  0&  0&  0.4&0.6\\ \hline
    \end{tabular}
   
\end{table}

A fundamental difference in the behavior of the generators in the two difference market paradigms is the presence of forecast energy requirement constraint, which in the EMIR model forces the generator to provide either energy imbalance reserve or day-ahead energy. The generator must in some way participate in the day-ahead market, which  creates variability in its profits across the different scenarios. In scenarios where fuel costs, real time deviations, and closeout costs are low, the generator realizes positive profits, but when those same costs are higher it incurs losses. For risk-averse generators, the primary way to reduce losses in the unfavorable scenarios is to invest in advanced fuel, which provides a hedge by lowering fuel costs in scenarios when intraday costs exceed the fixed cost of advanced fuel. This investment smooths profits across scenarios and reduces the generator’s exposure to downside risk. Thus, risk aversion drives greater fuel investment in the EMIR model compared to the EMO model, where no forward commitment is enforced and profits remain constant at zero across scenarios.

However, the observation that advanced fuel investment is greater in the EMIR model is not universally true, it critically depends on the strike price. For example, take a risk-averse generator ($\alpha_1 = 0.5$) and three different strike prices, K = \$5, \$10, and \$15. It's fuel investment and risk-adjusted probabilities are shown in Table \ref{tab:kvarysinglegen}. As the strike price increases, both fuel investment and the emphasis it places on its typically 'risky' scenario decreases. This is because the source of its risk, the option closeout cost, is decreasing as the strike price grows. Additionally, there is a preference for the generator to provide energy imbalance reserve instead of day-ahead energy. As K increases, the generator's benefit of providing day-ahead energy over reserve decreases.  Consider how much a generator would earn from choosing to provide day-ahead energy instead of reserve. It earns the price of energy and reserve and faces a deviation cost on each unit of DA energy, $\lambda^{DA} + \rho -\lambda^{RT}_s$, and forgoes the reserve clearing price and closeout charge, $\rho -[\lambda^{RT}_s-K]^+$. When the closeout charge is positive, the costs of providing energy over reserve reduce to  $\lambda^{DA} - K$, so as K increases above the day-ahead price, one can expect to see the overall amount of EIR awarded rise and day-ahead energy cleared fall, which is evident in the $g_i^{DA}$ and $e_i$ columns of Table \ref{tab:kvarysinglegen}. We state this more precisely below. 

{\bf{Observation:}} Note that the EIR mechanism provides a bona fide mechanism to mitigate risk and improve the generator's risk adjusted position.
In particular, we can compare the profits attached to the scenarios of importance under risk aversion. For such a scenario $s$, the profit received by the generator in the EMO context is given by 
$$
\begin{array}{c}
  Z_{i,s} = \lambda^{DA}g^{DA}_i +\lambda^{RT}_s(g^{RT}_{i,s}-g^{DA}_i)- \\ C_ig^{RT}_{i,s}- C^{F}_iV_{i}^{DA} - C^{I}_{i,s}V_{i,s}^{RT}.\\
 \end{array}
$$

The same choices of decision variables would result in a feasible solution for the EMIR model. However, note that given the possibility of investing in EIRs, the generator can offer (and be dispatched) $e_i = \epsilon$ for EIRs and reduce their day-ahead generation to $g_i^{DA} - \epsilon$. For a high enough value of $K$, and small enough $\epsilon$, the generator would have an incentive to make this change. Consider the day-ahead profits (in scenario $s$). This would change from $(\lambda^{DA} - \lambda^{RT})g_i^{DA}$ to 
$(\lambda^{DA} - \lambda^{RT})g_i^{DA} + \rho g_i^{DA} - (\lambda^{DA} - \lambda_s^{RT} -\rho -[\lambda_s^{RT}-K]^+) \epsilon.$ (Note that the real time profits would remain the same). Clearly for a proper choice of $\epsilon >0$, the latter profits are greater than the former hence the generator will choose to sell EIRs and improve their risk adjusted profits. 

\begin{table}
 \caption{Single Generator EMIR Model Solutions Under Different K's}
    \label{tab:kvarysinglegen}
    \centering
    \begin{tabular}{|c|c|c|c|c|c|c|c|} \hline 
        K & $V_1^{DA}$  &$g^{DA}_1$ &$e_1$& $Z_{1,1}$ & $Z_{1,2}$ & $q_{1,1}$ &$q_{1,2}$\\ \hline 
         5&  90&90 &0&  0&  0&  0.13&0.87\\ \hline 
         10&  75 &64.18&25.82&  0&  0&  0.16&0.84\\ \hline 
         15&  0&0&93.02&  0&  0&  0.5&0.5\\ \hline
    \end{tabular}
   
\end{table}

\subsection{Two Generators Example}

In this section, we present results from a numerical example with two competitive generators and five real time scenarios. The demand side agent participates in the day-ahead market in this example. Input parameters are presented in Tables \ref{tab:2genins1} and \ref{tab:2genins2}. The purpose of these numerical results, similar to the previous section, is to explore the differences in solutions under different market structures, namely the generator investment decisions, consumer costs, day-ahead demand, etc. as we vary market structures, risk-attitudes, and, when relevant,  K values. 
\begin{table}
 \caption{Two Generator Inputs}
    \label{tab:2genins1}
    \centering
    \begin{tabular}{|c|c|c|c|c|c|} \hline 
         Scenario&  1&  2&  3&  4& 5\\ \hline 
         $D^{RT}_s$& 50 &75  &  100&  125& 150\\ \hline 
         $C^I_{i,s}$&  15&  20&  30&  50& 100\\ \hline 
          $R_{i,s}$&  10&  10&  10&  10& 10\\ \hline 
 $\pi_s$& 0.2& 0.2& 0.2& 0.2&0.2\\ \hline
    \end{tabular}
   
\end{table}
\begin{table}
  \caption{Two Generator Inputs}
    \label{tab:2genins2}
    \centering
    \begin{tabular}{|c|c|c|c|} \hline 
         &  $Q_i$&  $C_i$&  $C^{F}_i$\\ \hline 
         Gen 1&  100&  0&  50\\ \hline 
         Gen 2&  100&  5&  50\\ \hline
    \end{tabular}
  
\end{table}

\subsubsection{Relationship Between $\alpha_i$ and $V_i^{DA}$}
Figure \ref{riskvfuel} shows aggregate fuel investment for both generators under a range of $\alpha_i$ values from 1 to 0.1 (in increments of 0.1) for the models with and without EIR. $\alpha_i$ is the same for both generators. The strike price is constant for all $\alpha_i$ and is set at value of K = \$50/MWh. In the EMIR model, there is a trend for generators to purchase more advanced fuel as they become more risk averse. At $\alpha_i$ levels of 1, 0.9, and 0.8, the generators purchase no advanced fuel in both models. When $\alpha_i$ 0.7 or lower, the aggregate fuel purchase in the EMIR model increases until it reaches a level equal to the load forecast. The advanced fuel investment in the EMO is zero for all $\alpha_i$ values, for similar reasons to why the single generator does not invest in advanced fuel in the simple example shown in \ref{tab:singlegenemo}. The generators are not required to participate in the day-ahead energy market, which reduces their exposure to risk, and allows them to achieve sufficient risk adjusted profits without investing in advanced fuel. 
\begin{figure}
         \centering
         \includegraphics[scale=.6]{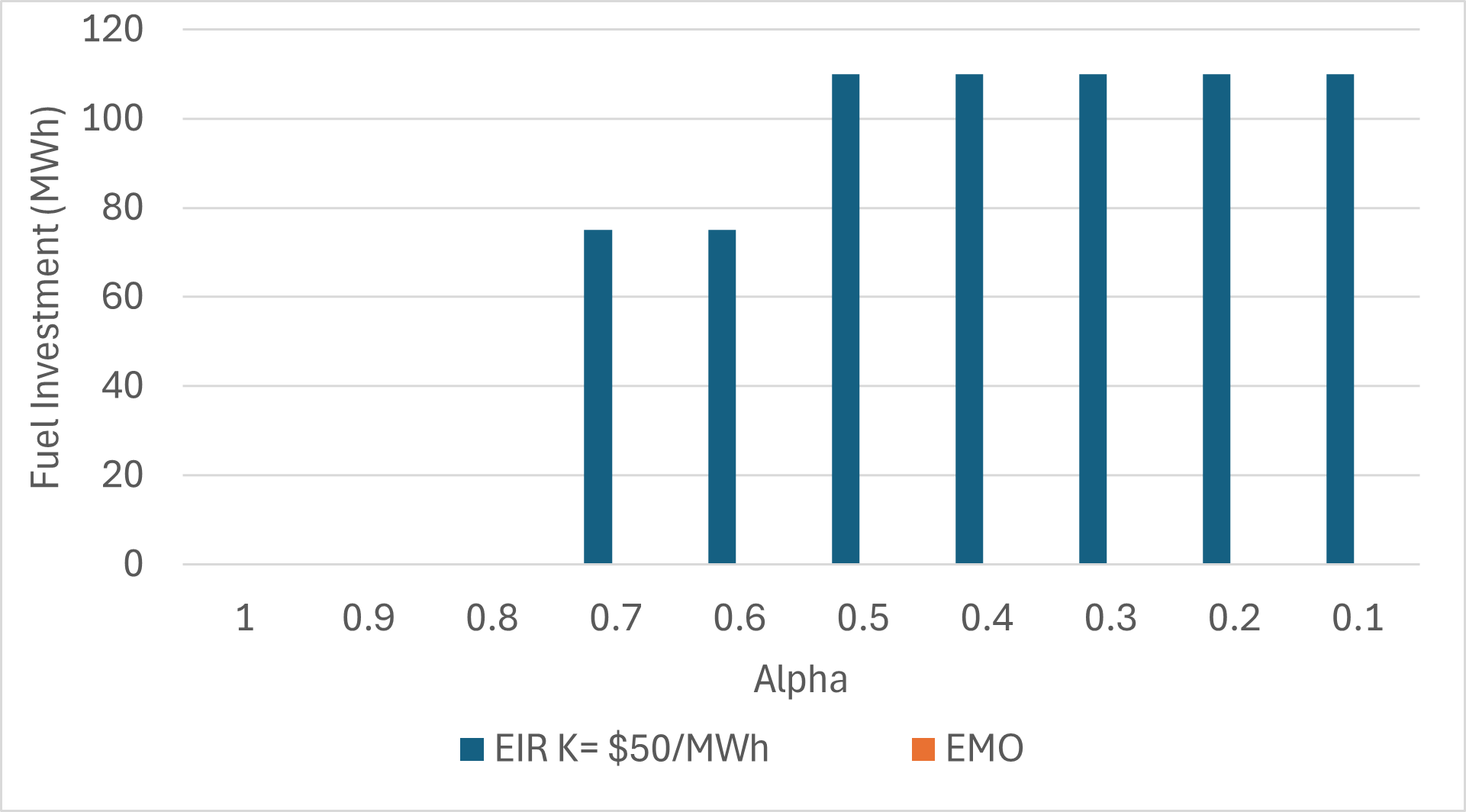}
         \caption{Advanced Fuel Investment vs Alpha when K = \$50/MWh }
         \label{riskvfuel}
\end{figure}
\subsubsection{Relationship Between K and $V_i^{DA}$/$e_i$/$g_i^{DA}$}
The investment in advanced fuel changes with respect to the strike price. Figure \ref{kvsfuel} shows aggregate fuel investment for two generators with $\alpha_i$ = 0.6 and a strike price, K, that varies from \$0 to \$100/MWh. Observe that as the strike price increases, the aggregate fuel investment decreases. As K increases, the size of the closeout cost diminishes, furthermore, this means the corresponding risk adjusted probability of a payout scenario also decreases. From the perspective of the generator,  when K is high there is limited risk associated with providing EIR and less incentive for them to change their behavior as result. When K $\le$ \$60/MWh, aggregate advanced fuel investment is equal to the forecasted energy requirement. This relationship highlights the importance of ISO-NE's decision in setting the strike price as it is the tool to provide generators with proper incentives to invest in advanced fuel. In our example, once $K>$ \$90/MWh, generators utilize the spot market only to backup their real time generation. 

\begin{figure}
         \centering
         \includegraphics[scale=.65]{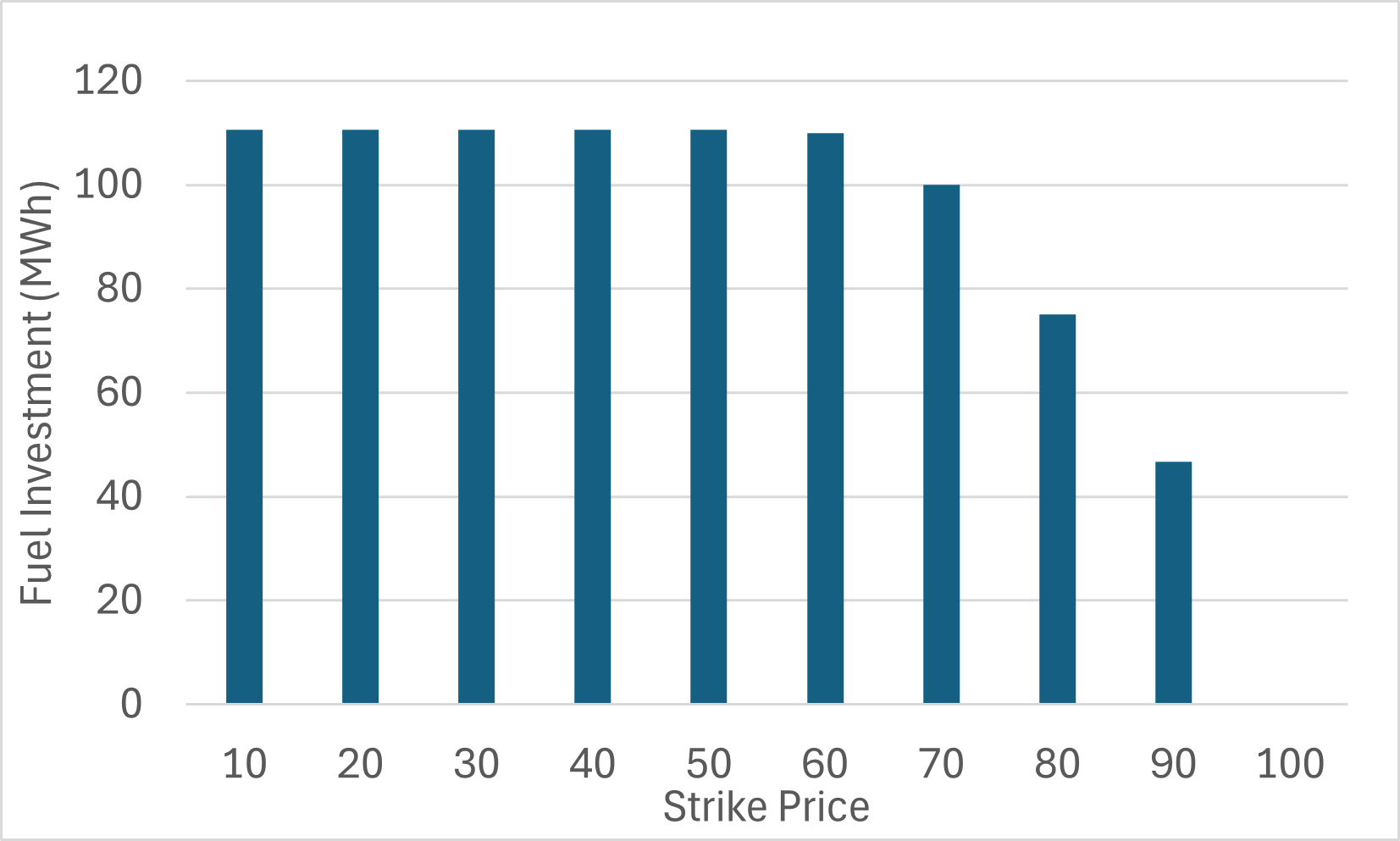}
         \caption{Advanced Fuel Investment vs Strike Price When $\alpha_{1}$=$\alpha_2$ = 0.6 in EMIR Model }
         \label{kvsfuel}
\end{figure}
It is notable that the fuel investment is identical for all strike prices below \$60/MWh, which is close to \$20 above the expected real time price. This shows that, in this example with the assumption that both agents' $\alpha_i = 0.6$, the incentives for fuel investment are not greatly impacted when the strike price is slightly higher than the expected real time price.
Changing the strike price impacts the ratio of day-ahead generation and energy imbalance reserve awarded in order to satisfy the FER constraint. Figure \ref{kvsdaeir} shows the total day-ahead and EIR quantities under a range of K values. As K increases, generators offer higher amounts of EIR relative to day-ahead generation. At the lower end of the strike price range, generators provide only $g_i^{DA}$ and no EIR is traded. When K $>$ \$90, FER is mostly met through energy imbalance reserve. This is relevant because from a producer's perspective providing EIR when K is larger results in a smaller per unit cost compared to day-ahead energy. EIR is charged $[\lambda^{RT}_s - K]^+$, whereas DA energy is charged $-\lambda^{RT}_s$.  The smaller costs from providing EIR relative to day-ahead generation allows the generator to achieve sufficient profits without having to invest in day-ahead fuel, resulting in diminished $V_i^{DA}$ as shown in Figure \ref{kvsfuel}. 

\begin{figure}
         \centering
         \includegraphics[scale=.6]{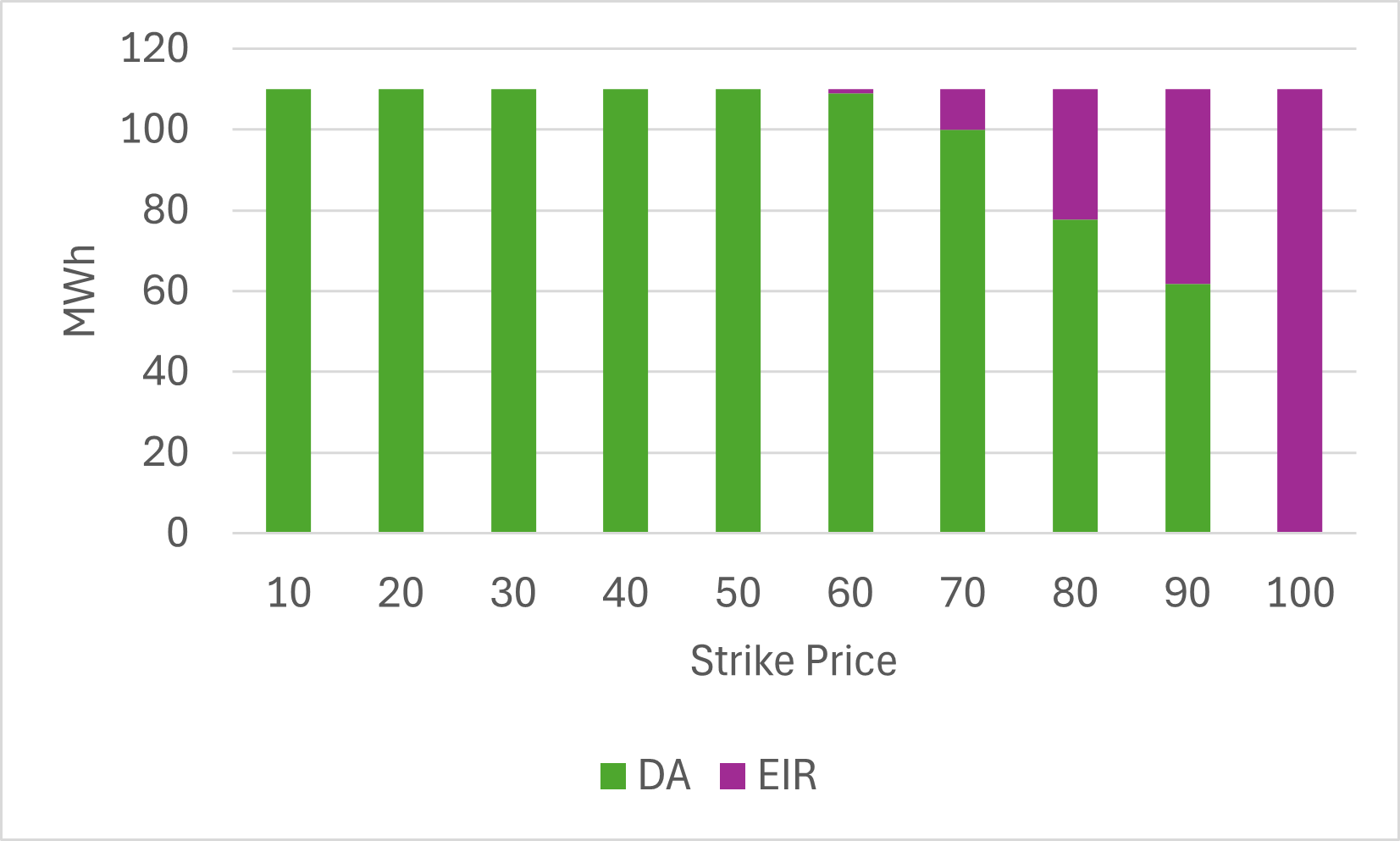}
         \caption{Day-ahead and EIR Awards vs Strike Price When $\alpha_{1}$=$\alpha_2$ = 0.6 }
         \label{kvsdaeir}
\end{figure}

\subsubsection{Risk Averse Demand Agent}
All previous results involved a risk-neutral demand agent. When the demand agent becomes risk-averse, it increases its day-ahead demand purchases above and beyond the load forecast in both the EMIR and EMO models. Table \ref{tab:ddarad} shows physical day-ahead purchases across a range of $\alpha^D$ values when generators are risk neutral. Across all models, physical demand purchases day-ahead energy in excess of the highest real time demand scenario at every $\alpha_i$ value below 1. In other words, a risk averse demand agent buys more energy on the DA market than it will consume in real time regardless of whether EIR is present or not.
\begin{table}[H]
\caption{$d^{DA}$ with Risk Averse Demand Agent}
    \label{tab:ddarad}
    \centering
    \begin{tabular}{|c|c|c|} \hline 
 & EIR& EMO\\ \hline 
         1&  
61.4305
&  3.00217
\\ \hline 
         .9&  
166.541
&  
166.665
\\ \hline 
         .8&  
166.541
&  
166.665
\\ \hline 
         .7&  
166.541
&  166.665
\\ \hline 
         .6&  175.818
&  
175.861
\\ \hline 
         .5&  172.222
&  172.222
\\ \hline 
         .4&  172.222
&  172.222
\\ \hline 
         .3&  172.222
&  172.222
\\ \hline 
         .2&  172.059
&  172.059
\\ \hline 
         .1&  172.059&  172.059\\ \hline
    \end{tabular}
    
\end{table}

This occurs for a range of different strike prices, as shown in Tables \ref{tab:ddavsk}.

\begin{table}[h]
    \caption{$d^{DA}$ vs K with Risk Averse Demand($\alpha^D = 0.5)$}
    \label{tab:ddavsk}
    \centering
    \begin{tabular}{|c|c|}\hline
 K&$d^{DA}$\\\hline \hline 
         10& 172.222
\\ \hline 
         20& 172.222
\\ \hline 
         30& 172.222
\\ \hline 
         40& 172.222
\\ \hline 
         50& 172.222
\\ \hline 
         60& 172.222
\\ \hline 
         70& 172.222
\\ \hline 
         80& 172.222
\\ \hline 
         90& 172.222
\\ \hline 
         100& 172.222\\ \hline
    \end{tabular}
    
\end{table}

\section{Comparison Between Energy Imbalance Reserve Market and Energy Only Market with Load Forecast Constraint}
The introduction of EIR is one possible policy choice to incentivize advanced fuel investment and raise day-ahead energy to the level of the system operator's load forecast. Consider another choice, where the system operator simply constrains the day-ahead generation to be greater than or equal to the load forecast, $\sum_i g_i^{DA} \ge FER$. Many system operators do in fact schedule units ahead of real time using a load forecast in place of demand bids, in a an out-of-market process that goes by different names across ISO/RTO regions (Reserve Adequacy Analysis in ISO-NE, Reliability Unit Commitment in ERCOT/PJM, Residual Unit Commitment in CAISO, etc. \cite{ecco_international_design_nodate}). In this section, we examine what would occur if the load forecast constraint enforced in the out of market reliability commitment process was brought into the day-ahead market and compare this to the EMIR model results. The full MCP formulation of the energy market with load forecast constraint is presented in Appendix IX.B.
\subsection{Equilibrium Model for Energy Market Only with Load Forecast Constraint}
\subsubsection{Generators' Profit Maximization Problem} 
The individual generator's optimization problem is to maximize their risk adjusted profits. The generator's problem without EIR products is presented below with the associated dual variables in brackets to the right of the constraints. The term, $\lambda^{LF}$, is a price paid to physical generators that assist in satisfying the FER constraint on top of the $\lambda^{DA}$. $\lambda^{LF}$ is the shadow price of the load forecast constraint shown in \ref{DALF}.

\begin{subequations}
    \begin{align}
        \max_{\Xi} & \textcolor{white}{h}\eta_i- \frac{1}{a_i}\sum_s \pi_su_{i,s}\\
\mbox{s/t}&\textcolor{white}{h} u_{i,s} \ge \eta_i - Z_{i,s}&[q_{i,s}] \\
&g^{DA}_i \le Q_i &[\delta_i] \\
& g^{RT}_{i,s} \le Q_i &[\gamma_{i,s}]\\
& g^{RT}_{i,s} = F_i + V_{i}^{DA}  + V_{i,s}^{RT}- W_{i,s}^{RT} &[\mu_{i,s}]\\
& g^{RT}_{i,s},g^{DA}_i,V^{DA}_i,V^{RT}_{i,s},W^{RT}_{i,s},u_{i,s} \ge 0 &  \forall s 
    \end{align}
\end{subequations}
where
\begin{align}
\Xi := (g^{RT}_{i,s},g^{DA}_i,u_{i,s},V^{DA}_i,V^{RT}_{i,s},W^{RT}_{i,s},\eta_i,u_{i,s})\nonumber\\
  Z_{i,s} := (\lambda^{DA}+\lambda^{LF})g^{DA}_i +\lambda^{RT}_s(g^{RT}_{i,s}-g^{DA}_i)-\label{Gen_CVARprofit_EMOLF} \\ C_ig^{RT}_{i,s}- C^{F}_iV_{i}^{DA} - C^{I}_{i,s}V_{i,s}^{RT} + R_{i,s}W_{i,s}^{RT} \nonumber
 \end{align}
\subsubsection{Demand's Profit Maximization Problem in EMO}\label{EMO_demand_profit_section} 
In this section, we provide a demand agent that represents the aggregate of all demand side market participants and pays for all energy in EMO. $d^{DA}$ is the sole decision variable of this agent, as we assume their real time loads are fixed and represented by $D^{RT}_s$. Demand agent's problem in absence of EIR options can be described as follows:

\begin{subequations}
    \begin{align}
        \underset{{\eta^D,u_s^D,d^{DA}}}{\mbox{max}} & \eta^D - \frac{1}{\alpha^D}\sum{\pi_s}u_{s}^D &  \\
\mbox{s/t}&\textcolor{white}{y} u_{s}^D \ge \eta^D - Z_{s}^D&[q_{s}^D]& \\
& u_{s}^D,d^{DA} \ge 0 & \forall s
    \end{align}
\end{subequations}
where
\begin{equation}
  Z_{s}^D := -\lambda^{DA}d^{DA}-\lambda^{LF}FER-\lambda^{RT}_s(D^{RT}_{s}-d^{DA}) \label{dem_CVARprofit_EMOLF} \\
\end{equation}

\subsubsection{Arbitragers' Profit Maximization Problem in EMO} \label{EMO_arbitragers_profit_section} 
Arbitragers' problem is identical the formulation shown in VI-A3.

\subsubsection{Market Clearing Conditions in EMO}
The previous day-ahead and real time market clearing conditions are present in this formulation. 
\begin{equation}
0 \le  \sum_i g^{DA}_i - FER  \perp \lambda^{LF} \ge 0 \label{DALF}
 \end{equation}
\subsection{Numerical Results}
Using the same inputs from the previous numerical results with a single generator and two generators, shown in Tables \ref{tab:onegeninputs},\ref{tab:2genins1} and \ref{tab:2genins2}, we compare the outcomes of the models with model to those of the EMO with load forecast. 
\subsubsection{Single Generator and Two RT Scenarios Example}
 Table \ref{tab:lfsingle} shows the solutions for this case for different $\alpha_i$ values.  We observe a similar trend to the case with EIR; as the generator becomes more risk-averse it invests in more advanced fuel, its profits across the two scenarios converge to zero, and it places greater weight on the second scenario. However, an important difference is that these results are solely due to the presence of the load forecast constraint. This suggests that structural changes to the market—such as mandatory forward scheduling based on system forecasts—may offer a viable alternative or complement to financial reserve products like the EIR. Such approaches could deliver similar reliability and investment incentives without the need for complex instruments.
\begin{table}[h]
\caption{Single Generator EMO with Load Forecast Solutions}
    \label{tab:lfsingle}
    \centering
    \begin{tabular}{|c|c|c|c|c|c|c|} \hline 
        $\alpha_1$ & $V_1^{DA}$  &$g^{DA}_1$ & $Z_{1,1}$ & $Z_{1,2}$ & $q_{1,1}$ &$q_{1,2}$\\ \hline 
         1&  0 &90 &  225&  -225&  0.5&0.5\\ \hline 
         0.7&  75&90&  75&  -30&  0.28&0.72\\ \hline 
         0.4&  90&90&  0&  0&  .13&.86\\ \hline
    \end{tabular}
    
\end{table}

\subsubsection{Two Generators and Five RT Scenarios Example}
Figure \ref{riskvfuelloadforecast} shows aggregate fuel investment for both generators under a range of $\alpha_i$ values from 1 to 0.1 (in increments of 0.1) for the models with and without EIR. The strike price is constant for all $\alpha_i$ and is set at value of K = \$50/MWh. In both the EMIR and EMO with a load forecast constraint, there is a trend for generators to purchase more advanced fuel as they become more risk averse. At $\alpha_i$ levels of 1, 0.9, and 0.8, the less risk-averse generators purchase either none or negligible amounts of advanced fuel advance because the value of fuel investment is less than the cost of advanced fuel investment, $\sum_s \mu_{i,s} < C^F_i$. When $\alpha_i$ is less than or equal to 0.5, the aggregate fuel purchase is the same in both models and is equal to the load forecast. At the middle levels of risk ($\alpha_i = 0.8-0.6$), there is a difference between the level of aggregate fuel investment across the two models, with the largest discrepancy when $\alpha=0.7$ resulting in 25 MWh more fuel investment in the EMO with a load forecast. 
\begin{figure}[H]
         \centering
         \includegraphics[scale=.6]{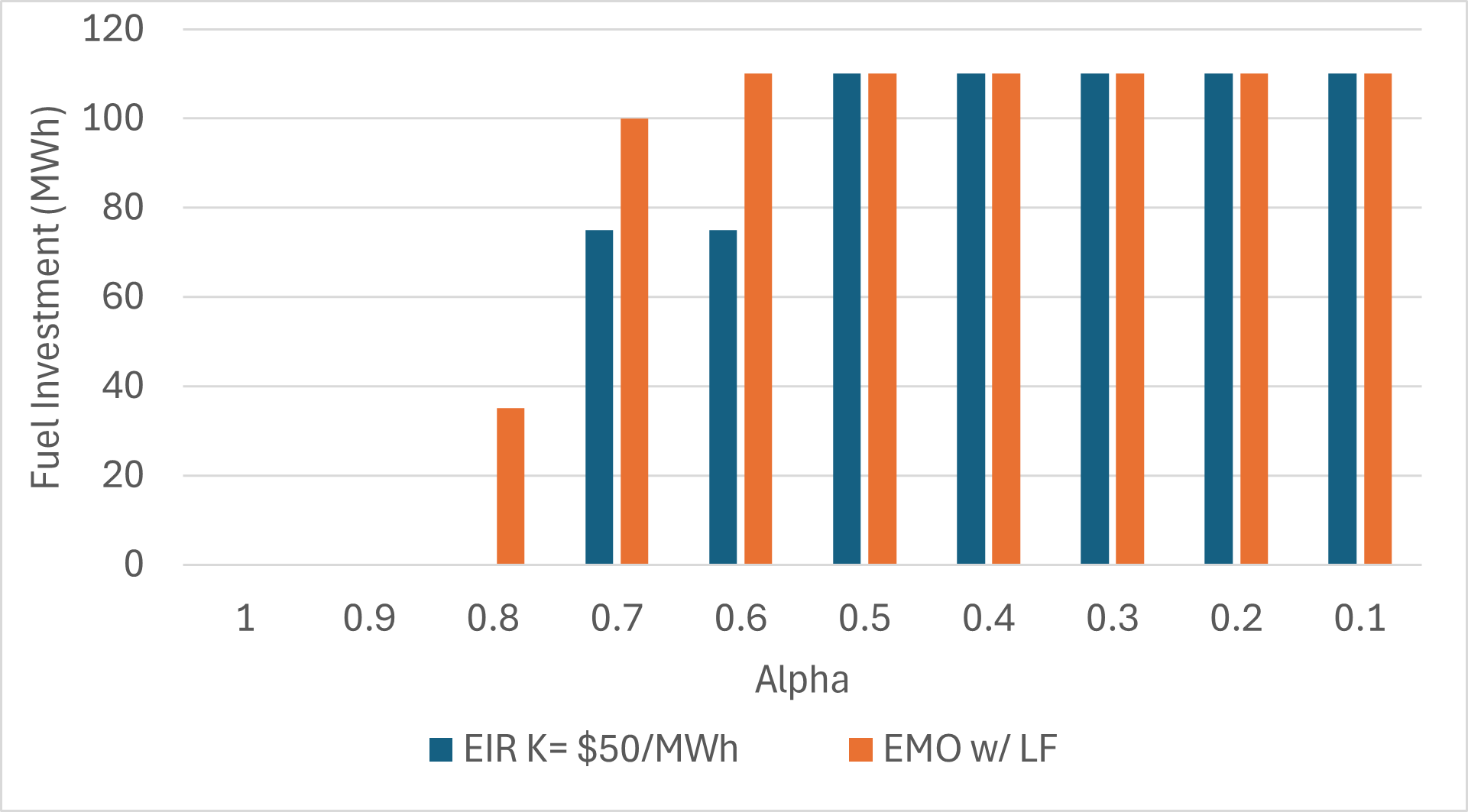}
         \caption{Advanced Fuel Investment vs Alpha when K = \$50/MWh }
         \label{riskvfuelloadforecast}
\end{figure}

\section{Discussion}

In this paper, we introduced, modeled, and analyzed the potential effects of ISO New England’s Energy Imbalance Reserve (EIR) product by developing simple equilibrium models that incorporate risk-averse producers and consumers, both with and without the EIR mechanism. We model risk using conditional value at risk, a coherent risk measure with convexity assumptions. Similarly, our simplified models are convex, ignoring unit commitment and other intertemporal constraints. Our analysis revealed that under risk-neutral assumptions, outcomes in the two market designs—EMIR and EMO—were nearly indistinguishable. In the risk neutral case, the presence of the EIR had negligible impact on supply or demand side behavior, indicating that the reserve product’s effectiveness depends on market participants’ risk attitudes.

When risk aversion was introduced, especially on the supply side, notable differences emerged between the two paradigms. Risk-averse generators in the EMIR model were more inclined to undertake higher-cost investments in advanced fuel to mitigate their risk. These investments were particularly sensitive to the strike price parameter K. As K increased, the incentive to invest in advanced fuel declined due to the lower option closeout charge, while the likelihood of EIR provision increased. This trade-off underscores the importance of carefully choosing K to balance incentives for both reliability and cost-effective reserve provision.

On the demand side,  risk-averse consumers exhibited increased participation in the day-ahead market across both EMIR and EMO models, suggesting that risk aversion alone, and not the presence of EIR—can drive larger day-ahead commitments from load.

While these findings offer valuable insights, translating them directly into policy recommendations should be done with caution. The EIR market is already operational, and empirical evaluations by independent market monitors will be crucial in assessing its real-world impacts. Our study was forward-looking and aimed at providing a theoretical lens on possible behavioral shifts prior to market implementation. Nonetheless, two key takeaways emerge from our work. Future work could extend this analysis by incorporating additional features, such as strategic behavior by market participants or transmission constraints to see how a nodal vs system wide strike price can impact economic incentives. Empirical studies comparing real-world EIR performance against results of this paper will  be instrumental in refining modeling approaches going forward.

\section{Acknowledgment}
This work was funded by the Independent System Operator of New England and authors Tongxin Zheng and Jinye Zhao are employees of the Advanced Technology Solutions department of ISO-NE. The authors would like to thank Andrew Withers, Zeky Murra-Anton, Parviz Alivand and Geoffrey Pritchard for their informative discussions and comments. 
\section{Appendix}

\scriptsize
\subsection{Risk Neutral Models}
EMO
\begin{align}
0 \le C_i\pi_s + \mu_{i,s} - \lambda^{RT}_s\pi_s + \gamma_{i,s} & \perp g^{RT}_{i,s} \geq 0 \label{RT1dem}\\
0 \le Q_i - g^{RT}_{i,s} & \perp \gamma_{i,s} \ge 0 \label{RT2dem} \\
g^{RT}_{i,s} = F_i + V_{i,s}^{RT} + V_{i}^{DA} - W_{i,s}^{RT} & \perp  \mu_{i,s} \mbox{ free} \label{RT3dem}\\
0 \le C^{F}_i - \sum_s\mu_{i,s} & \perp V_{i}^{DA} \ge 0 \label{DA1dem} \\
0 \le C^{I}_{i,s}\pi_s - \mu_{i,s} & \perp V_{i,s}^{RT} \ge 0 \label{RTFdem}\\
0 \leq \mu_{i,s} - R_{i,s}\pi_s& \perp W_{i,s}^{RT} \geq 0 \label{wconstrdem}\\
0 \le \sum_s \pi_s\lambda^{RT}_{s} - \lambda^{DA}  + \delta_i & \perp g^{DA}_i \ge 0 \label{DA2dem}\\ 
0 \le Q_i - g^{DA}_i & \perp \delta_i \ge 0 \label{DACapdem}\\
0 \le \sum_s \pi_s\lambda^{RT}_s - \lambda^{DA} & \perp a \ge 0 \label{ARBT1dem} \\
0 \le \lambda^{DA} - \sum_s \pi_s\lambda^{RT}_s & \perp b \ge 0 \label{ARBT2dem} \\
\sum_i g^{DA}_i + a = d^{DA} + b & \perp \lambda^{DA} \mbox{ free} \label{MC1dem}\\
\sum_i g^{RT}_{i,s} = D^{RT}_s  & \perp \lambda^{RT}_s \mbox{ free} \label{MC2dem}\\
   0 \le \lambda^{DA} - \sum_s \pi_s\lambda^{RT}_s & \perp d^{DA}\ge 0  \label{DDA1dem} 
\end{align}
EMIR
\begin{align}
0 \le C_i\pi_s + \mu_{i,s} - \lambda^{RT}_s\pi_s + \gamma_{i,s} & \perp g^{RT}_{i,s} \ge 0 \label{eirRT1dem}\\
0 \le Q_i - g^{RT}_{i,s} & \perp \gamma_{i,s} \ge 0  \label{eirRT2dem} \\
g^{RT}_{i,s} = F_i + V_{i}^{DA} + V_{i,s}^{RT} - W_{i,s}^{RT} & \perp  \mu_{i,s} \mbox{ free} \label{eirRT3dem} \\
0 \le C^{F}_i - \sum_s \mu_{i,s} & \perp V_{i}^{DA} \ge 0 \label{eirDAFdem}\\
0 \le C^{I}_{i,s}\pi_s - \mu_{i,s} & \perp V_{i,s}^{RT} \ge 0 \label{eirRTFdem}\\
0 \leq \mu_{i,s} - R_{i,s}\pi_s & \perp W_{i,s}^{RT} \geq 0 \label{eirwconstrdem}\\
0 \le \sum_s \pi_s\lambda^{RT}_{s} - \lambda^{DA} - \rho + \delta_i & \perp g^{DA}_i \ge 0 \label{eirDA1dem}\\ 
0 \le Q_i - g^{DA}_i - e_i & \perp \delta_i \ge 0 \label{eirDA2dem} \\
0 \le \sum_s \pi_s({\lambda}^{RT}_{s} - K)^+ - \rho +\delta_i & \perp e_i \ge 0 \label{eirEIRdem} \\
0 \le \sum_s \pi_s\lambda^{RT}_s - \lambda^{DA} & \perp a \ge 0 \label{eirARBT1dem} \\
0 \le \lambda^{DA} - \sum_s \pi_s\lambda^{RT}_s & \perp b \ge 0 \label{eirARBT2dem}\\
\sum_i g^{DA}_i + a = d^{DA} + b & \perp \lambda^{DA} \mbox{ free} \label{eirMC1dem}\\
\sum_i g^{RT}_{i,s} = D^{RT}_s  & \perp \lambda^{RT}_s \mbox{ free} \label{eirMC2dem}\\
0 \le \sum_i g^{DA}_i + \sum_i e_i - FER  & \perp \rho \ge 0 \label{eirMC3dem}\\
0 \le \lambda^{DA} - \sum_s \pi_s\lambda^{RT}_s & \perp d^{DA}\ge 0  \label{eirDDA1dem} 
\end{align}
\scriptsize{
\printbibliography
}
\newpage

\vfill

\end{document}